\documentclass[a4paper,12pt,USenglish,cleveref, autoref, thm-restate]{article}
\usepackage{graphicx} %
\usepackage{todonotes}
\usepackage{xfrac}
\usepackage{tikz}
\usepackage[normalem]{ulem}
\usepackage{amsmath,amsfonts,amsthm}
\usepackage{url}
\usepackage{hyperref}
\usepackage{thmtools} 
\usepackage{thm-restate}
\usepackage{xspace}
\usepackage{subcaption}
\usepackage{pdfpages}
\usepackage[top=30mm, bottom=40mm, left=30mm, right=30mm]{geometry}

\DeclareTextFontCommand{\str}{\ttfamily\upshape}

\newtheorem{theorem}{Theorem}
\newtheorem{example}[theorem]{Example}
\newtheorem{lemma}[theorem]{Lemma}

\newtheorem{remark}{Remark}
\newtheorem{conjecture}[theorem]{Conjecture}

\DeclareMathOperator{\GF}{\mathbb{F}}
\DeclareMathOperator{\ord}{ord}
\newcommand{\BWT}{\textsc{BWT}\xspace}
\DeclareMathOperator{\DB}{DB}
\DeclareMathOperator{\rle}{\textsc{BWT}\xspace}
\newcommand{\Z}{\mathbb{Z}}
\newcommand{\N}{\mathbb{N}}
\newcommand{\redstr}[1]{\textcolor{red}{\str{#1}}}
\newcommand{\bluestr}[1]{\textcolor{blue}{\str{#1}}}

\title{Unclustered BWTs of any Length\\ over Non-Binary Alphabets}
\author{Gabriele Fici, Estéban Gabory,\\Giuseppe Romana, Marinella Sciortino}
\date{August 2025}

\begin{document}

\maketitle

\begin{abstract}
We prove that for every integer $n > 0$ and for every alphabet $\Sigma_k$ of size $k \geq 3$, 
there exists a necklace of length $n$ whose Burrows--Wheeler Transform (BWT) 
is completely unclustered, i.e., it consists of exactly $n$ runs with no two consecutive 
equal symbols. These words represent the worst-case behavior of the BWT for clustering, since the number of BWT runs is maximized. We also establish a lower bound on their number. This contrasts with the binary case, where the existence of infinitely many completely unclustered BWTs is still an open problem, related to Artin's conjecture 
on primitive roots. 

\end{abstract}

\section{Introduction}
The Burrows–Wheeler Transform (BWT) is a transformation on words introduced in 1994 \cite{BW94}. When applied to a text, it produces a permutation obtained by concatenating the last character of each cyclic rotation (also called shift) of the text, after sorting all such rotations in lexicographic order. 
Since two shifts of the same word yield the same BWT, this transform can be naturally viewed as a mapping from the set of necklaces over a finite alphabet $\Sigma_k=\{0,1,\ldots, k-1\}$ (i.e., equivalence classes of words under shift operation) to the set of words of $\Sigma_k^*$.
The BWT is a widely used tool in data compression and indexing, because it is very likely to produce long runs of identical consecutive symbols (BWT-runs) when the input text is highly repetitive (clustering effect). This property is effectively exploited when combined with run-length encoding techniques. However, it is interesting to note that, while the BWT of a word typically exhibits some degree of clustering, there exist texts for which no such clustering occurs.

In this work, we focus on the worst-case scenario: words whose BWT is completely unclustered, meaning that no two consecutive characters are the same. That is, the BWT consists of exactly as many runs as the length of the word.  These unclustered BWTs also represent the worst cases for compression and indexing algorithms that use the Run-Length BWT.

This problem was previously studied by Mantaci et al. \cite{mantaci_measuring_2017} in the case of binary alphabets, revealing a connection with the generators of the multiplicative group $\Z^*_{n}$ of integers modulo $n$, as stated in the following result:

\begin{theorem}[\cite{mantaci_measuring_2017}]\label{thm:binary}
There exists a $2n$-length necklace with completely unclustered \BWT if and only if $2n+1$ is an odd prime and $2$ generates the cyclic group $\Z^*_{2n+1}$.
\end{theorem}

However, the question of whether infinitely many binary words with completely unclustered BWT exist is left open in \cite{mantaci_measuring_2017}. This problem is tightly connected to the still-open
Artin’s conjecture on the existence of primitive roots modulo infinitely many prime numbers.

Here, we move beyond the binary setting and prove that the existence of completely unclustered BWTs is guaranteed for every length as soon as the alphabet contains at least three letters. In particular, we prove the following:

\begin{restatable*}{theorem}{mainresult}\label{thm:main-result}
For every integer $n>0$, and $k\ge 3$, there exists a necklace $u$ of length $n$ over the alphabet $\Sigma_k$ having completely unclustered \BWT. 
\end{restatable*}

Our approach builds upon the structure of generalized de~Bruijn words and their correspondence with Hamiltonian paths in generalized de~Bruijn graphs introduced in~\cite{fici2025generalizedbruijnwordsinvertible}. In particular, we prove the following result:

\begin{restatable*}{theorem}{unclusteredDBW}\label{thm:unclusteredDBW}
For every integer $n>0$, and $k\ge 3$, there exists a generalized de Bruijn word of length $kn$ over $\Sigma_k$ with completely unclustered \BWT. In particular, there exists a (standard) de Bruijn word of length $k^n$ over $\Sigma_k$ with completely unclustered \BWT.
\end{restatable*}

The key point is that generalized de~Bruijn words, defined in the sense of~\cite{fici2025generalizedbruijnwordsinvertible}, are necklaces having a BWT consisting of consecutive blocks each of which is a permutation of the alphabet. As proved in~\cite{fici2025generalizedbruijnwordsinvertible}, generalized de~Bruijn words are precisely the Hamiltonian cycles of the generalized de~Bruijn graphs, introduced in the early $80$'s independently by Imase and Itoh~\cite{DBLP:journals/tc/ImaseI81}, and by Reddy, Pradhan, and Kuhl~\cite{Reddy} (see also~\cite{DBLP:journals/networks/DuH88}) in the context of network design.

Here, we provide a new algorithmic strategy that exploits combinatorial properties of generalized de~Bruijn graphs, making it possible 
to remove equal letters between successive blocks of the BWT, thereby obtaining necklaces of length $kn$ whose BWTs are completely unclustered. We then show how to extend this procedure to every 
length, not just multiples of the size of the alphabet.

Beyond existence, we also investigate counting aspects: we provide a lower bound of $\Omega(2^n/n)$ on the number of ternary necklaces of length $3n$ with completely unclustered \BWT. Finally, in Section~\ref{sec:Artin} we examine a special case connected to Artin’s conjecture, showing that certain highly structured words are \BWT images precisely when $kn+1$ is prime and $k$ is a primitive root modulo $kn+1$.

\section{Preliminaries}

We begin by introducing some preliminary definitions. For a thorough introduction, we refer the reader to \cite{crochemore_algorithms_2007} and \cite{lothaire_combinatorics_1983}.

Let $\Sigma_k=\{0,1,\ldots,k-1\}$, $k>1$, be a sorted set of \emph{letters}, called the \emph{alphabet}. 
  
A \emph{word} over the alphabet $\Sigma_k$ is a concatenation of elements of $\Sigma_k$. The \emph{length} of a word $w$ is denoted by $|w|$. %
 For a letter $i\in\Sigma_k$, $|w|_i$ denotes the number of occurrences of $i$ in $w$. The vector $(|w|_{0},\ldots,|w|_{k-1})$ is the \emph{Parikh vector} of $w$.

Let  $w=w_0\cdots w_{n-2}w_{n-1}$ be a word of length $n>0$. When $n>1$, the \textit{shift} of $w$ is the word $\sigma(w)=w_{n-1}w_0\cdots w_{n-2}$.

For a word $w$, the \emph{$n$-th power} of $w$ is the word $w^n$ obtained by concatenating $n$ copies of $w$. We call an \emph{alphabet-permutation}  a word over $\Sigma_k$ that contains each letter of $\Sigma_k$ exactly once, and an \emph{alphabet-permutation power} a concatenation of one or more alphabet-permutations over $\Sigma_k$.

\begin{example}\label{ex:alph_permut}
    The word $w=210201102102120$ is an alphabet-permutation power over $\Sigma_3$.
\end{example}

Two words $w$ and $w'$ are \emph{conjugates} if $w=uv$ and $w'=vu$ for some words $u$ and $v$. The conjugacy class of a word $w$ can be obtained by repeatedly applying the shift operator, and contains $|w|$ distinct elements if and only if $w$ is \emph{primitive}, i.e., for any nonempty word $v$ and integer $n$, $w= v^n$ implies $n=1$. %
A \emph{necklace} (resp.~\emph{aperiodic necklace}) $[w]$ is a conjugacy class of words (resp.~of \emph{primitive words}). Necklaces are also called \emph{circular words}. 

For example, $[1100]=\{1100,0110,0011,1001\}$ is an aperiodic necklace, while $[1010]=\{1010,0101\}$ is not aperiodic.

 The \textit{Burrows--Wheeler matrix} (BWT matrix) of a necklace $[w]$ is the matrix whose rows are the $|w|$ shifts of $w$ in ascending %
 lexicographic order. Let us denote by $F$ and $L$, respectively, the first and the last column of the BWT matrix of $[w]$. We have that $F=0^{n_0}1^{n_1}\cdots (k-1)^{n_{k-1}}$, where $(n_0,\ldots, n_{k-1})$ is the Parikh vector of $[w]$; $L$, instead, is the \emph{Burrows--Wheeler Transform} (BWT) of $[w]$. The BWT is therefore a map from the set of necklaces to the set of words. As is well known, it is an injective map. A word is a \emph{BWT image} if it is the BWT of some necklace. 

The \emph{run-length encoding} of a word $w\in\Sigma^n$, denoted $\rle(w)$, is the sequence $(w_i,\ell_i)_{\{1\le i \le r\}}\in(\Sigma\times \N)^r$, where $w=w_1^{\ell_1}\dots w_r^{\ell_r}$, and $w_i\neq w_{i+1}$ for every $1\le i \le r$. If $w=\BWT(u)$ for a necklace $u$, we write $\rho(u)=r$. We say that $u$ has a \emph{completely unclustered \BWT} if $\rho(u)=|u|$.

The \emph{standard permutation}  of a word $u=u_0u_1\cdots u_{n-1}$,  $u_i\in\Sigma_k$, is the permutation $\pi_u$ of $\{0,1,\ldots,n-1\}$ such that $\pi_u(i)<\pi_u(j)$ if and only if $u_i<u_j$ or  $u_i=u_j$ and $i<j$. In other words, in one-line notation, $\pi_u$ orders distinct letters of $u$ lexicographically, and equal letters by occurrence order, starting from $0$. %
When $u$ is the BWT image of some necklace, the standard permutation is also called $LF$-mapping. %

The \emph{inverse standard permutation} $\pi^{-1}_u$ of a word $u$, written in one-line notation, can be obtained by listing in left-to-right order the positions of $0$ in $u$, then the positions of $1$, and so on. The inverse standard permutation of a BWT image is also called $FL$-mapping.

\begin{lemma}[Folklore]\label{lem:bwt-cycle}
A $n$-length word $w$ is a $\BWT$ of some aperiodic necklace $u$ if and only if $\pi_u$ (or equivalently, $\pi^{-1}_u$) is a cycle of length $n$.
\end{lemma}

\begin{example}\label{ex:cyclic}
Let us consider the word $w=210201102102120$ used in Example \ref{ex:alph_permut}. The standard permutation and its cycle decomposition are as follows: %
\begin{align*}
\pi_w &= 
\left(
\begin{array}{ccccccccccccccc}
0 & 1 & 2 & 3 & 4 & 5 & 6 & 7 & 8 & 9 & 10 & 11 & 12 & 13 & 14 \\
10 & 5 & 0 & 11 & 1 & 6 & 7 & 2 & 12 & 8 & 3 & 13 & 9 & 14 & 4
\end{array}
\right) \\
&= (0, 10, 3, 11, 13, 14, 4, 1, 5, 6, 7, 2)(8, 12, 9)
\end{align*}
By using Lemma \ref{lem:bwt-cycle}, $w$ is not the image of any aperiodic necklace under the BWT. 
\end{example}

 Let $ G = (V, E)$ be a finite directed graph, which may have loops and multiple edges (in this case, it is also called a multigraph in the literature). Each edge
$e \in E$ is directed from its source vertex $s(e)$ to its target vertex $t(e)$. We write $indeg(v)=|\{e\in E,~ t(e)=v\}|$, and $outdeg(v)=|\{e\in E,~ s(e)=v\}|$ for every $v\in V$.

The directed \emph{line graph} (or \emph{edge graph}) $\mathcal{L} G =(E, E')$ of $G$ has as vertices the edges of $G$, and as edges the set
$$E' = \{(e_1, e_2) \in E \times E \mid s(e_2) = t(e_1)\}.$$

A directed graph is \emph{Hamiltonian} if it has a \emph{Hamiltonian cycle}, i.e., one that traverses each node exactly once, while it is \emph{Eulerian} if it has an \emph{Eulerian cycle}, i.e., one that traverses each edge exactly once. As is well known, a directed graph is Eulerian if and only if $indeg(v) = outdeg(v)$ for all vertices $v$. If a graph is Eulerian, its line graph is Hamiltonian.

Recall that the de Bruijn graph $\DB(k,k^d)$ of order $d$ is the directed graph whose vertices are the words of length $d$ over $\Sigma_k$ and there is an edge from $a_iw$ to $v$, where $a_i$ is a letter, if and only if $v=wa_j$ for some letter $a_j$. As is well known, de Bruijn graphs are Eulerian and Hamiltonian. One has $\mathcal{L}\DB(k,k^d) = \DB(k,k^{d+1})$. Now, since the vertex labels of $\DB(k,k^d)$ are the base-$k$ representations of the integers in $\{0,1,\ldots,k^d-1\}$ on $d$ digits (padded with leading zeroes), the de Bruijn graph of order $d$ over $\Sigma_k$ can be equivalently defined as the one having vertex set $\{0,1,\ldots,k^d-1\}$ and containing an edge from $m$ to $km+i \mod (k^d)$ for every $i=0,1,\ldots,k-1$.

The \emph{generalized de Bruijn graph} $\DB(k,n)$ has vertices $\{0,1,\ldots,n-1\}$ and for every vertex $m$ there is an edge from $m$ to $km+i \mod (n)$ for every $i=0,1,\ldots,k-1$. Generalized de Bruijn graphs have been introduced independently by Imase and Itoh~\cite{DBLP:journals/tc/ImaseI81}, and by Reddy, Pradhan, and Kuhl~\cite{Reddy} (see also~\cite{DBLP:journals/networks/DuH88}). They are Eulerian and Hamiltonian. 
The line graph of $\DB(k,n)$ is $\DB(k,kn)$~\cite{DBLP:journals/dm/LiZ91}. %
More precisely, one can show the following: 

\begin{lemma}[\cite{DBLP:journals/dm/LiZ91}]\label{lem:line graph}
Given $k,n$ integers, the graph $\DB(k,kn)$ is isomorphic to the line graph of $\DB(k,n)$, where the node corresponding to an edge $(i,ki+j)\mod n$ gets label $ki+j\mod kn$
\end{lemma}
The Eulerian cycles of $\DB(k,n)$ correspond to the Hamiltonian cycles of $\DB(k,kn)$.%

\begin{figure}
    \centering
    \includegraphics[width=0.7\linewidth]{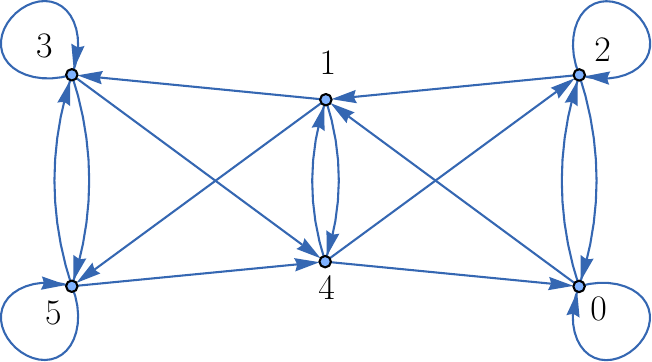}
    \caption{Generalized de Bruijn graph $\DB(3,6)$. The graph has six vertices, labeled $0,1,2,3,4,5$, each representing an integer modulo $6$. 
From every vertex $m$, there are directed edges to vertices $(3m+i) \bmod 6$ 
for each $i \in \{0,1,2\}$.}
    \label{fig:DB(3,6)}
\end{figure}

A \emph{de Bruijn word} of order $d$ over $\Sigma_k$ is a necklace over $\Sigma_k$ containing as a circular factor each of the $\Sigma_k^d$ words of length $d$ over $\Sigma_k$ exactly once. De Bruijn words correspond to Hamiltonian cycles in $\DB(\Sigma_k,k^d)$, as they can be obtained by concatenating the first letter of the labels of the vertices it traverses.

De Bruijn words were generalized in several different ways (see, e.g.,~\cite{blum_comparison_1983,au_generalized_2015,nellore_arbitrary-length_2022}). We use here the notion from~\cite{fici2025generalizedbruijnwordsinvertible}:
A necklace $u$ of length $kn$ over the alphabet $\Sigma_k$ is a \emph{generalized de Bruijn word} if $\BWT(u)$ is an alphabet-permutation power. Generalized de Bruijn words are aperiodic necklaces.

\begin{example}
    The necklace $v=[011021001120212220]$ is a generalized de Bruijn word of length $18$. In fact, $BWT(v)=210012210012210021$ is an alphabet-permutation power. Note that $\rho(v)=13$.
\end{example}

Observe that not all alphabet-permutation powers are BWT images, as shown in Example \ref{ex:cyclic}.

The following theorem establishes a direct correspondence between generalized de Bruijn words and Hamiltonian paths in generalized de Bruijn graphs.

\begin{theorem}[\cite{fici2025generalizedbruijnwordsinvertible}]\label{thm:dbw-graph}
A necklace $u$ of length $kn$ over $\Sigma_k$ is a generalized de Bruijn word if and only if $\pi_{\BWT(u)}^{-1}$, as a cycle, corresponds to a Hamiltonian path of $\DB(k,kn)$.
\end{theorem}

The next example shows a generalized de Bruijn word whose BWT is a completely unclustered alphabet-permutation power, providing also an instance of the statement of  Theorem \ref{thm:dbw-graph}.

\begin{example}\label{ex:stand-perm}
The necklace $u=[220011021002211201]$ is a generalized de Bruijn word of length $3\cdot 6=18$, since $BWT(u)=210210210210210210$, which is an alphabet-permutation power. Moreover, $\rho(u)=18=|u|$.
Note that the standard permutation $\pi_{BWT(u)}$ is a cycle of length $18$, 
$$\pi_{BWT(u)}=(0, 12, 16,11,3,13,10,9,15,17,5,1,6,14,4,7,8,2).$$

The inverse permutation is: 
$$\pi_{BWT(u)}^{-1} = (0,2,8,7,4,14,6,1,5,17,15,9,10,13,3,11,16,12)$$

\end{example}

\begin{figure}[htp]
  \centering
\includegraphics[width=\textwidth]{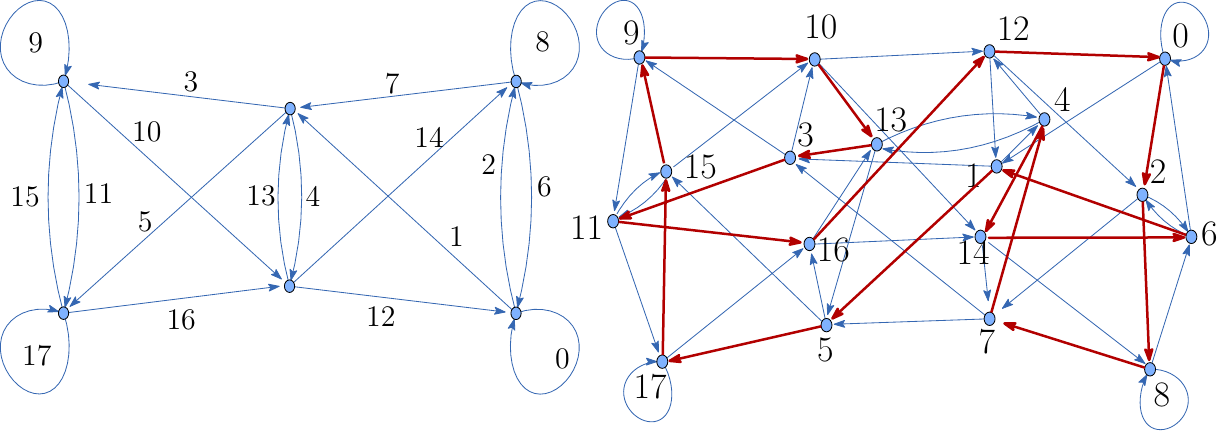}
  \caption{The graph $\DB(3,6)$ (left) with edge labels, and the graph $\mathcal{L} \DB(3,6)=\DB(3,18)$ (right), with node labels. The Hamiltonian path corresponding to $\pi^{-1}_{\BWT(u)}$ (Example~\ref{ex:stand-perm}) is highlighted in red on $\DB(3,18)$. It corresponds to an Eulerian path on $\DB(3,6)$. Notice that the edges of the form $(i,3i+j)\mod 6$ in $\DB(3,6)$ (where the nodes are labeled as in Figure~\ref{fig:DB(3,6)}) are labeled $3i+j\mod 18$ (Lemma~\ref{lem:line graph}).}
  \label{fig:permutation-to-graph}
\end{figure}

\section{Main result}
In this section, we establish the existence of necklaces with completely unclustered BWTs over alphabets of size at least three. Our approach is based on generalized de Brujin words and their representation as Hamiltonian paths in generalized de Brujin graphs. We first show how to construct such necklaces of length that is a multiple of the alphabet size, and then how to extend the construction to arbitrary lengths.

\subsection{Constructing completely unclustered BWTs from generalized de Bruijn words}

In this section, we prove that for every integer $n>0$, and $k\ge 3$, there exists a generalized de Bruijn word of length $\ell=kn$ over $\Sigma_k$ with completely unclustered \BWT. 

The key insight is that adjacent equal letters between consecutive alphabet-permutation blocks can be systematically eliminated by rerouting Hamiltonian cycles in the corresponding de Brujin graph. This yields generalized de Brujin words with completely unclustered \BWT.

Let $u$ be a generalized de Bruijn word, and let us write $w=\BWT(u)$. By definition, $w$ is an alphabet-permutation power, hence a sequence of length-$k$ blocks that are permutations of $\Sigma_k$. Clearly, a repeated letter can occur only at the boundary between two successive such blocks. In other words, $u$ has a completely unclustered \BWT if and only if $w_{ki+(k-1)}\neq w_{k(i+1)}$ for every $0\le i\le n-1$. On the contrary, if $w_{ki+(k-1)}= w_{k(i+1)}$, we say that $w$ has a \emph{tie} at block $i$. Ties for $w$ can also be characterized by the inverse standard permutation of $w$, as shown in the following lemma: 

\begin{lemma}\label{lem:tie charac}
Let $u$ $\in \Sigma_k^*$ be a generalized de Bruijn word of length $kn$, $w=\BWT(u)$, and $\sigma=\pi_w^{-1}$.
The word $w$ has a tie at block $i$ if and only if for some $j\in \{0,\dots, k-1\}$, one has:

\begin{equation}
    \begin{aligned}
  & \sigma(jn+i)=ki+k-1\\
&\sigma(jn+i+1)=k(i+1)
\end{aligned}
\end{equation}
\end{lemma}
\begin{proof}
    Let $u$ be such a necklace, and let us write $w=\BWT(u)$. 
    Consider the permutation $\sigma=\pi^{-1}_w$. It can be understood as follows: writing the multiset of letters of $w$ in lexicographic order, where equal letters are sorted by occurrence order in $w$, the position in $w$ of the $m$-th letter is $\sigma(m)$. In particular, since the Parikh vector of $w$ is $(n,\dots,n)$, the positions of $\str{0}$ in $w$ are $\sigma(0),\dots, \sigma(n-1)$, the positions of $\str{1}$ are $\sigma(n),\dots, \sigma(2n-1)$, and so on, until the last letter, which occurs at positions $\sigma((k-1)n),\dots, \sigma(kn-1)$. In particular, there is a tie in $w$ at block $i$ if and only if, for some $j\in\{0,\ldots,k-1\}$ the $j$th letter of $\Sigma_k$ occurs at positions $ki+k-1$ and $k(i+1)$, and in that case they are respectively the $(i+1)$-th and the $i+2$-th occurrences of this letter. In other words, we have $\sigma(jn+i)=ki+k-1$ and $\sigma(jn+i+1)=k(i+1)$, as claimed.
\end{proof}

The following lemma provides conditions under which Hamiltonian paths can be rerouted to avoid a specified edge. We will use it to break ties in BWTs of generalized de Bruijn words by rerouting the corresponding Hamiltonian paths so that they avoid one of the edges responsible for the tie.

\begin{lemma}\label{lem:Hamiltonian-rerouting}
Let $G$ be a directed Hamiltonian graph, and $A,B\subseteq V(G)$ such that: \begin{enumerate}
    \item both $A$ and $B$ have cardinality at least $3$;
    \item every pair $(x,y)\in A\times B$ is an edge of $G$;
    \item every edge starting at a node in $A$ ends at a node in $B$.
\end{enumerate}

Then, if $\gamma$ is a Hamiltonian cycle for $G$ visiting an edge $e=(x,y)\in A\times B$, it can be rerouted into a Hamiltonian cycle $\gamma'$ such that:

\begin{enumerate}
    \item $\gamma'$ does not visit $e$;
    \item if $e'\in E(G)\setminus (A \times B)$, $\gamma'$ visits $e'$ if and only if $\gamma$ visits $e'$.%
\end{enumerate}
\end{lemma}

\begin{proof}
Let $G$ be a directed Hamiltonian graph, $A$ and $B$ two subsets of $V(G)$ satisfying the conditions of the statement, and $\gamma$ a Hamiltonian cycle visiting $e=(x,y)\in A \times B$. Since $\gamma$ is a cycle, we can write it as a sequence of nodes starting with $y$ and ending with $x$. Sets $A$ and $B$ have cardinality at least $3$, and every node from $A$ must be followed by a node from $B$. Hence, one can write $\gamma=\gamma_1\cdot \gamma_2 \cdot \gamma_3$ where each $\gamma_i$ starts with a node from $B$ and ends with a node from $A$, for $i=1,2,3$. Since each $(x,y)\in A \times B$ is an edge, the sequence $\gamma'=\gamma_1 \cdot \gamma_3 \cdot \gamma_2$ is still a Hamiltonian cycle for $G$, which does not visit $(x,y)$. Furthermore, $\gamma$ and $\gamma'$ differ only on edges in $A\times B$, and the statement follows.
\end{proof}

\begin{remark}
Note that Lemma~\ref{lem:Hamiltonian-rerouting} does not require having $A\cap B = \emptyset$. 
\end{remark}

\begin{remark}
    If $G$ is the line graph of an Eulerian graph $\tilde{G}$, Lemma~\ref{lem:Hamiltonian-rerouting} has a simpler interpretation: the set $A$ (resp.~$B$) corresponds to the set of all incoming (resp.~all outgoing) edges for a fixed node $v$ having indegree and outdegree at least $3$, and the edge $e$ to be avoided corresponds to a sequence of two edges $(x,v)$ and $(v,y)$. In that case, the cycle $\gamma$ corresponds to an Eulerian cycle on $\tilde{G}$, which by definition must visit $v$ at least $3$ times, and can be written as a concatenation of $3$ cycles starting and ending at $v$. In that case, the rerouting performed simply corresponds to a permutation in the order in which those cycles are traversed. We gave the Hamiltonian version of the result because the correspondence of Theorem~\ref{thm:dbw-graph} is more direct in the Hamiltonian viewpoint, but all related statements could be reformulated in terms of Eulerian graphs.
\end{remark}

\begin{lemma}\label{lem:untie}
Let $u$ be a generalized de Bruijn word of length $kn$ over $\Sigma_k$, with $k\ge 3$, and $w=\BWT(u)$. Assume that $i_0$ is the rightmost (resp.~leftmost) tie for $w$. Then, we can construct a word $w'$ such that: (i) $w'=\BWT(u')$ for some generalized de Bruijn word $u'$; (ii) $w'_i=w_i$ for every $i$ outside of block $i_0$; and (iii) $w'$ has no tie at block $i$ for every $i\ge i_0$ (resp.~for every $i\le i_0$).
\end{lemma}
\begin{proof}
We start with an arbitrary generalized de Bruin word $u$ of length $\ell=kn$ over $\Sigma_k$. Let $w=\BWT(u)$. Let $i_0$ be the rightmost tie for $w$. One has $w_{ki_0+k-1}= w_{k(i_0+1)}$, and from Lemma~\ref{lem:tie charac}, we can fix $j_0\in\{0,\dots k-1\}$ such that $\sigma(j_0n+i_0)=ki_0+k-1$ and
$\sigma(j_0n+i_0+1)=k(i_0+1)$, where $j_0$ depends on the letter $w_{ki_0+k-1}= w_{k(i_0+1)}$.

From Lemma~\ref{lem:bwt-cycle}, the permutation $\sigma=\pi^{-1}_w$ is a cycle, and by Theorem~\ref{thm:dbw-graph}, this cycle corresponds to the consecutive node labels of a Hamiltonian cycle of $G=\DB(k,kn)$. Consider the sets $A_{i_0}=\{i_0,n+i_0,\dots kn+i_0\}\subseteq V(G)$, $B_{i_0}=\{ki_0,ki_0+1,\dots ki_0+k-1\}\subseteq V(G)$. Note that the edge $(j_0n+i_0+1,k(i_0+1))$ is visited by the cycle corresponding to $\sigma$ and that it is an edge from $A_{i_0}$ to $B_{i_0}$. Those sets both have cardinality $k\ge 3$, and by construction of $G$, they satisfy the conditions of Lemma~\ref{lem:Hamiltonian-rerouting}. We can then obtain a permutation $\sigma'$, also corresponding to an Hamiltonian cycle of $G$, with $\sigma'(j_0n+i_0)\neq ki_0+k-1$, hence $w'$ has no tie at $i_0$, and $\sigma'(m)=\sigma(m)$ for every $m\not\in A_{i_0}$, hence the words $w$ and $w'$ differ only at block $i_0$ (namely, at positions $ki_0\dots ki_0+k-1$). Since $i_0$ was the rightmost tie in $w$, we deduce that $w'$ has no tie at block $i\ge i_0$. %

If $i_0$ is the leftmost tie for $w$, then we can apply Lemma~\ref{lem:Hamiltonian-rerouting} with sets $A_{i_0+1}$ and $B_{i_0+1}$ to obtain a permutation $\sigma'$ such that $\sigma'(j_0n+i_0+1)\neq k(i_0+1)$, hence breaking the tie at block $i_0$ by modifying only block $i_0+1$, and obtaining a word $w'$ that is the \BWT of a generalized de Bruijn word and has no tie at block $i\le i_0$. %
\end{proof}

We can now prove our main result:

\unclusteredDBW

\begin{proof}
By Theorem~\ref{thm:dbw-graph}, generalized de Bruijn words are in bijection with Hamiltonian cycles of $\DB(k,kn)$, or equivalently, with Eulerian cycles of $\DB(k,n)$. Since $\DB(k,n)$ is Eulerian for every $k,n$, there is at least one generalized de Bruijn word $u$ of length $kn$ over $\Sigma_k$. If $\BWT(u)$ has no tie, then $u$  is a generalized de Bruijn with completely unclustered \BWT. Otherwise, by repeating the process of Lemma~\ref{lem:untie} on $w=\BWT(u)$, we eventually obtain, by induction, a string $w'$ with no ties, which is then the \BWT of some generalized de Bruijn word $u'$ that has a completely unclustered \BWT. 
\end{proof}

\begin{example}\label{ex:untie}
Let $u=[\str{001022112}]$. One has $w=\BWT(u)=\str{201021120}$, hence $u$ is a generalized de Bruijn word with $k=3$ and $n=3$. Since $w$ has a tie at block $i=1$, $u$ does not have a completely unclustered \BWT. One has $\sigma=\pi^{-1}_w=(0, 1, 3, 2, 8, 7, 4, 5, 6)$, %
and $\sigma(4)=\sigma(n+i)=5=k\cdot i+2$, and $\sigma(5)=\sigma(n+i+1)=6=k\cdot (i+1)$. We have $A_1=\{i,n+i,2n+i\}=\{1,4,7\}$, $B_1=\{3i,3i+1,3i+2\}=\{3,4,5\}$, and we write $\gamma=560132874=\gamma_1\cdot \gamma_2 \cdot \gamma_3$, with $\gamma_1=5601$, $\gamma_2=3287$, $\gamma_3=4$. Then, $\gamma'=\gamma_1 \cdot \gamma _3 \cdot \gamma_2=560143287$ also corresponds to a cycle on $\DB(3,9)$, namely $\sigma'=(0,1,4,3,2,8,7,5,6)$, and $\sigma'=\pi^{-1}_{w'}$, with $w'=\str{201102120}=\BWT(u')$, $u'=[\str{001102212}]$. The string $w'$ still has one tie---observe, however, that the tie is now at block $i'=0$, and that there are no ties at later blocks. 

One has $\sigma'(3)=\sigma'(n+0)=2=k\cdot i+2$, and $\sigma'(4)=\sigma(n+1)=3=k\cdot (i+1)$. We have $A_0=\{0,3,6\}$, $B_0=\{0,1,2\}$, and we write $\gamma^{(2)}=287560143=\gamma_1^{(2)}\cdot \gamma_2^{(2)}\cdot \gamma_3^{(2)}$ with $\gamma_1^{(2)}=28756$, $\gamma_2^{(2)}=0$ $\gamma_3^{(2)}=143$. Then, $\gamma'^{(2)}=\gamma_1^{(2)}\cdot \gamma_3^{(2)} \cdot \gamma_2^{(2)}=287561430$ corresponds to a cycle of $\DB(3,9)$, namely $\sigma'^{(2)}=(0,2,8,7,5,6,1,4,3)$, and $\sigma'^{(2)}=\pi^{-1}_{w'^{(2)}}$, with $w'^{(2)}=\str{120102120}=\BWT(u'^{(2)})$, $u'^{(2)}=[\str{002212011}]$. We finally get the word $u^*=u'^{(2)}$, which is a de Bruijn word with completely unclustered \BWT.
\end{example}

\begin{figure}[htp]
  \centering

  \begin{subfigure}[b]{0.3\textwidth}
    \includegraphics[page=1, width=\linewidth]{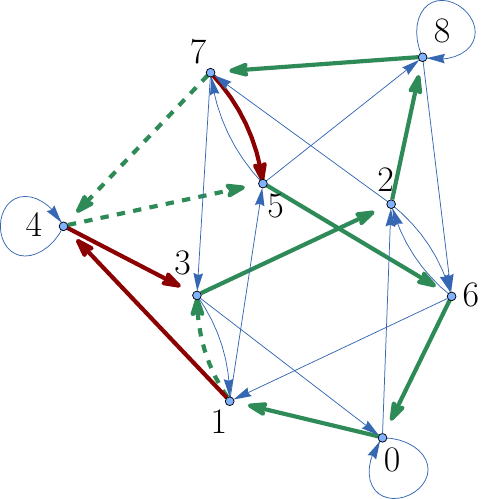}
  \end{subfigure}
  \hfill
  \begin{subfigure}[b]{0.3\textwidth}
    \includegraphics[page=2, width=\linewidth]{figures/figure-dbg.pdf}
  \end{subfigure}
  \hfill
  \begin{subfigure}[b]{0.3\textwidth}
    \includegraphics[page=3, width=\linewidth]{figures/figure-dbg.pdf}
  \end{subfigure}

  \caption{Illustration of the process on the de Bruijn graph $\DB(3,9)$, as described in Example~\ref{ex:untie}.  
In the left subfigure, the green cycle corresponds to the initial permutation $\sigma$, and the dashed green edges (resp.~red edges) indicate the edges to remove (resp.~to add) to obtain the updated permutation $\sigma'$.  
In the middle subfigure, the green cycle corresponds to $\sigma' = \sigma^{(2)}$, and the dashed green edges (resp.~red edges) indicate the edges to remove (resp.~to add) to obtain $\sigma'^{(2)}$.  
The right subfigure shows the final cycle $\sigma'^{(2)}$ in green, which corresponds to a generalized de Bruijn word with a completely unclustered \BWT.}
  \label{fig:three-subfigs}
\end{figure}

\subsection{Extending to all lengths}

After establishing the existence of completely unclustered BWTs for generalized de Brujin words of length $kn$, we extend the construction to arbitrary word lengths in the ternary case. In fact, we show how to insert or delete single letters while preserving the cyclic structure of the standard permutation. Finally, we extend the result to larger alphabets.

\begin{lemma}\label{lem:positionof2}
    For every integer $n>0$, there exists a necklace $u$ of length $n$ over $\Sigma_3$ having completely unclustered \BWT, written $w$, and such that $w_{3n-3}=\str{2}$ (resp.~such that $w_{3n-2}=\str{2}$).
\end{lemma}
\begin{proof}
Let $u$ be a ternary generalized de Bruijn word of length $3n$, and $w=\BWT(u)$. Let $\sigma=\pi^{-1}_w$ be its inverse standard permutation, corresponding to a Hamiltonian path $\gamma$ of $\DB(k,3k)$. Let $A_n=\{n-1,2n-1,3n-1\}$ and $B_n=\{3n-3,3n-2,3n-1\}$, and write $e_1$ for the edge $(3n-1,3n-2)$ and $e_2=(3n-1,3n-3)$. Notice also that the edge $(3n-1,3n-1)$ is visited by no Hamiltonian path, since it is a self-loop. Now, $\gamma$ must visit one edge among $e_1$ and $e_2$: in particular, it visits $e_1$ if and only if $w_{3n-2}=\str{2}$, and $e_2$ if and only if $w_{3n-3}=\str{2}$. In both cases, since $A_n$ and $B_n$ satisfy the conditions of Lemma~\ref{lem:Hamiltonian-rerouting}, we can obtain a Hamiltonian path $\gamma'$ visiting $e_2$ (resp.~$e_1$) if and only if $\gamma$ visits $e_1$ (resp.~$e_2$), hence a string $w'$ which is the $\BWT$ image of a de Bruijn word and such that $w'_{3n-3}=\str{2}$ (resp.~$w'_{3n-2}=\str{2}$) if and only if $w_{3n-2}=\str{2}$ (resp.~$w_{3n-3}=\str{2}$). We now apply Lemma~\ref{lem:untie} iteratively to each of the words $w$ and $w'$, following the same procedure as in the proof of Theorem~\ref{thm:unclusteredDBW}. At each step, we eliminate the rightmost tie (at some block $j \le n - 1$) until the word becomes completely unclustered. Crucially, this process does not affect the final block of the word, so the position of the letter $\str{2}$—whether at $w_{3n-3}$ or $w_{3n-2}$—remains unchanged. Hence, in both cases, we obtain a completely unclustered word $w$ (or $w'$) with a $\str{2}$ at the prescribed position, as desired.
\end{proof}

\mainresult

\begin{proof}
    The result clearly holds for $n = 1, 2, 3$, so assume $n \ge 4$. If $n \equiv 0 \mod 3$, the conclusion follows from Theorem~\ref{thm:unclusteredDBW}.

Assume now $n \equiv 1 \mod 3$, so $n = 3n' + 1$ for some $n'$. By Lemma~\ref{lem:positionof2}, there exists a ternary generalized de Bruijn word $u$ of length $3n'$, with completely unclustered $\BWT$ $w = \BWT(u)$, such that $w_{3n'-3} = \str{2}$. Let $w'$ be the word of length $3n'+1$ obtained by inserting a $\str{2}$ between positions $3n'-2$ and $3n'-1$. This insertion is not adjacent to the existing $\str{2}$ at position $3n'-3$, so it increases the number of runs by $1$, and $w'$ remains completely unclustered.

Let us verify that \(\pi_{w'}\), the standard permutation of \(w'\), is a cycle. Since the insertion occurs after the last $\str{2}$ of \(w\), all positions before $3n'-1$ preserve their relative order: for all $i < 3n'-1$, we have $\pi_{w'}(i) = \pi_w(i)$. The inserted $\str{2}$ at position $3n'-1$ is now the new last $\str{2}$, so $\pi_{w'}(3n'-1) = 3n'$. The former position $3n'-1$ (now shifted to $3n'$) keeps its image: $\pi_{w'}(3n') = \pi_w(3n'-1)$. All subsequent values are unchanged. Thus, in one-line notation, $\pi_{w'}$ is obtained by inserting $3n'$ just after $3n'-1$ in the cycle of $\pi_w$. Since $\pi_w$ is a $(3n')$-cycle, this operation yields a $(3n'+1)$-cycle, so $w' = \BWT(u')$ for some necklace $u'$ of length $n$, which is completely unclustered.

Now assume $n \equiv 2 \mod 3$, so $n = 3n' - 1$. Again by Lemma~\ref{lem:positionof2}, there exists a generalized de Bruijn word $u$ of length $3n'$, with completely unclustered $\BWT$ $w = \BWT(u)$ such that $w_{3n'-2} = \str{2}$. Let $w'$ be the word of length $3n'-1$ obtained by deleting the $\str{2}$ at position $3n'-2$. As the last block of $w$ is a permutation of $\Sigma_3$, we have $w_{3n'-3}\neq w_{3n'-1}$, so $w'$ remains completely unclustered.

We now track the effect on the standard permutation. For all $i < 3n'-2$, the relative order remains the same, so $\pi_{w'}(i) = \pi_w(i)$. Since the removed $\str{2}$ was the last of its kind, we had $\pi_w(3n'-2) = 3n'-1$. The symbol at position $3n'-1$ in $w$ is now at $3n'-2$ in $w'$, so $\pi_{w'}(3n'-2) = \pi_w(3n'-1)$. All subsequent values are preserved. Thus, $\pi_{w'}$ is obtained by deleting $3n'-1$ from the cycle of $\pi_w$, which yields a $(3n'-1)$-cycle. Therefore, $w' = \BWT(u')$ for some necklace $u'$ of length $n$, which is completely unclustered.
\end{proof}

\begin{example}

    Let $w = \str{201021120}$ as in Example~\ref{ex:untie}. It has length $9$, and removing ties yields $\hat{w} = \str{120102120} = \BWT([\str{002212011}])$. We have $w_7 = \hat{w}_7 = \str{2}$, so removing the last $\str{2}$ in $\hat{w}$ gives $\hat{w}' = \str{12010210}$, a completely unclustered word of length $8$, which is the BWT of $\hat{u}'=[\str{00212011}]$.

    On the other hand, we can construct from $w$ the word $\overline{w} = \str{201021201} = \BWT([\str{001021122}])$, which is already completely unclustered and satisfies $\overline{w}_6 = \str{2}$. Hence, inserting a $\str{2}$ in penultimate position gives $\overline{w}' = \str{2010212021}$, a completely unclustered word of length $10$, which is the BWT of $\overline{u}' = [\str{0010211222}]$.

\end{example}

Any ternary word $w\in\Sigma_3^*$ is a special case of a word over $\Sigma_k$, with $k>3$, where only the letters $0$, $1$, and $2$ occur. 
Hence, the result of Theorem~\ref{thm:main-result} can be easily extended to any alphabet of size greater than $3$.
It is natural to wonder whether the same result can be achieved if we force $w$ to contain every distinct letter at least once.

\begin{theorem} \label{th:karyunclustered}
    For every pair of integers $n\geq k\geq 3$, there exists a necklace $u$ of length $n$ over $\Sigma_k$, with $|u|_i\ge1$ for all $i\in\{0,\ldots,k-1\}$, having completely unclustered \BWT. 
\end{theorem}

\begin{proof}
    Let $w\in\Sigma_3^n$ be any word having $n$ runs that is a BWT image of some aperiodic necklace $u$, which exists by Theorem~\ref{thm:main-result}.
    We can then define a word $w'$ as follows:
    \begin{equation*}
        w'[\pi_w(n-i)]=\begin{cases}
        k-i & \text{ if $1\leq i\leq k-3$;}\\
        2   & \text{ if $i=k-2\wedge k-3\geq|w|_2$;}\\
        1   &\text{ if $i=k-1\wedge k-3\geq|w|_2+|w|_1$;}\\
        w[\pi_w(n-i)]  & \text{ otherwise.}
    \end{cases}
    \end{equation*}
    By construction, one can observe that $\pi_w = \pi_{w'}$, i.e., $w'$ is the BWT image of some aperiodic necklace $u'$ containing at least one occurrence of each character in $\Sigma_k$. Moreover, the word $w'$ is obtained by replacing some one-letter runs in $w$ with characters that occur only once in $w'$, thus preserving the property of being completely unclustered, and the thesis follows.
\end{proof}

\begin{example}\label{ex:unclusterdkaryalph}
    Consider the necklace $u=[\str{001021122}]$ and its completely unclustered BWT $w=\str{201021201}$ with standard permutation $\pi_w=(0,6,8,5,4,7,2,3,1)$. The word $w'=\str{301041502}\in \Sigma^*_6$, which is obtained according to the method described in the proof of Theorem~\ref{th:karyunclustered}, is the BWT of the necklace $u'=[\str{001041253}]$.
    A graphical representation is shown in Figure~\ref{fig:karayunclustered}. 
\end{example}

\begin{figure}
    \centering
        \begin{tabular}{ccccccccc}
            \str{0} &  \str{0} & \str{1} & \str{0} & \bluestr{2} & \str{1} & \bluestr{1} &  \bluestr{2} & \bluestr{2}\\
            \str{0} & \str{1} & \str{0} & \bluestr{2} & \str{1} & \bluestr{1} &  \bluestr{2} &  \bluestr{2} & \str{0}\\
            \str{0} & \bluestr{2} & \str{1} & \bluestr{1} &  \bluestr{2} &  \bluestr{2} & \str{0}  & \str{0} & \str{1}\\
            \str{1} & \str{0} & \bluestr{2} & \str{1} & \bluestr{1} &  \bluestr{2} &  \bluestr{2} & \str{0}  & \str{0}\\
            \str{1} & \bluestr{1} &  \bluestr{2} &  \bluestr{2} & \str{0}  &  \str{0} & \str{1} & \str{0}  & \bluestr{2}\\
            \bluestr{1} &  \bluestr{2} &  \bluestr{2} & \str{0} & \str{0} & \str{1} & \str{0}  & \bluestr{2} & \str{1}\\
            \bluestr{2} & \str{0} & \str{0} & \str{1} & \str{0}  & \bluestr{2} & \str{1} & \bluestr{1} & \bluestr{2}\\
            \bluestr{2} & \str{1} & \bluestr{1} &  \bluestr{2} &  \bluestr{2} & \str{0} & \str{0} & \str{1} & \str{0}\\
            \bluestr{2} &  \bluestr{2} & \str{0}  &  \str{0} & \str{1} & \str{0} & \bluestr{2} & \str{1}  & \bluestr{1}\\
        \end{tabular}
        \hspace{20pt}
        \begin{tabular}{ccccccccc}
            \str{0} &  \str{0} & \str{1} & \str{0} & \redstr{4} & \str{1} & \redstr{2} &  \redstr{5} & \redstr{3}\\
            \str{0} & \str{1} & \str{0} & \redstr{4} & \str{1} & \redstr{2} &  \redstr{5} & \redstr{3} & \str{0}\\
            \str{0} & \redstr{4} & \str{1} & \redstr{2} &  \redstr{5} &  \redstr{3} & \str{0}  & \str{0} & \str{1}\\
            \str{1} & \str{0} & \redstr{4} & \str{1} & \redstr{2} &  \redstr{5} &  \redstr{3} & \str{0}  & \str{0}\\
            \str{1} & \redstr{2} &  \redstr{5} & \redstr{3} & \str{0}  &  \str{0} & \str{1} & \str{0}  & \redstr{4}\\
            \redstr{2} &  \redstr{5} &  \redstr{3} & \str{0} & \str{0} & \str{1} & \str{0}  & \redstr{4} & \str{1}\\
            \redstr{3} & \str{0} & \str{0} & \str{1} & \str{0}  & \redstr{4} & \str{1} & \redstr{2} & \redstr{5}\\
            \redstr{4} & \str{1} & \redstr{2} &  \redstr{5} &  \redstr{3} & \str{0} & \str{0} & \str{1} & \str{0}\\
            \redstr{5} &  \redstr{3} & \str{0}  &  \str{0} & \str{1} & \str{0} & \redstr{4} & \str{1}  & \redstr{2}\\
        \end{tabular}
    \caption{BWT matrices for the necklaces $u=[\str{001021122}]$ (on the left) and $u' = [\str{001041253}]$ (on the right). The last column of each matrix correspond to their respective BWT's $w=\str{201021201}$ and $w'=\str{301041502}$. We highlight in blue (red) the letters of the BWT matrix of $u$ ($u'$) before (after) the substitution of the letters according to the construction of $w'$ in the proof of Theorem~\ref{th:karyunclustered}.} 
    \label{fig:karayunclustered}
\end{figure}

\section{Lower bounds}

For an alphabet of size $3$, we can obtain a lower bound on the number of words having completely unclustered BWT, using the \emph{generalized Euler's totient function}.

Let $p$ be a prime number. The generalized Euler's totient function counts the number of polynomials over $\GF_p$ of degree smaller than $n$ and coprime with $X^n-1$.
It can be computed using the formula
\begin{equation}\label{eq:euler}
  \Phi_p(n)=p^n\prod_{d | (n/\lambda_p(n))}\left(1-\dfrac{1}{p^{\ord_p(d)}}\right)^{\tfrac{\phi(d)}{\ord_p(d)}}  
\end{equation}

for $n>1$, 
where $\lambda_p(n)$ is the largest power of $p$ dividing $n$, and $\ord_p(d)$ is the multiplicative order of $d$ modulo $p$.

\begin{theorem}
    For every integer $n>0$, there exists at least $\Phi_3(n)/2n=\Omega(2^n/n)$ necklace of length $3n$ over $\Sigma_3$ having a completely unclustered $\BWT$.
\end{theorem}
\begin{proof}
The number of generalized de Bruijn words of length $3n$ over $\Sigma_3$ %
is $2^{n-1}\cdot \dfrac{\Phi_3(n)}{n}$~\cite{fici2025generalizedbruijnwordsinvertible}. By applying successively Lemma~\ref{lem:untie} to any such word, we eventually reach a completely unclustered \BWT. Each such operation ends up choosing a set $S\subseteq [0,\dots n-1]$ of blocks to untie. Since there are $2^n$ such subsets, each completely unclustered de Bruijn word can be constructed from at most $2^n$ generalized de Bruijn words; hence, we obtain at least $\Phi_3(n)/2n$ such words.

Now, from equation~\ref{eq:euler}, and by observing that $\ord_3(d)\ge 1$ for any $d$ and that $$\sum_{d | (n/\lambda_p(n))} \frac{\phi(d)}{\ord_p(d)}\le \sum_{d | (n/\lambda_p(n))} \phi(d)=n/\lambda_p(n),$$ we obtain $$\dfrac{\Phi_3(n)}{2n} \ge \frac{1}{2n}\cdot 3^n(1-\frac{1}{3})^{n/\lambda_p(n)}\ge \frac{1}{2n}\cdot 3^n(1-\frac{1}{3})^n\ge \frac{2^{n-1}}{n}.$$
\end{proof}

\section{Special case related to Artin's conjecture}\label{sec:Artin}

We showed that necklaces with completely unclusterized \BWT exist for every length when the alphabet has size at least $3$, contrasting with the binary case (Theorem~\ref{thm:binary}). In particular, this case is related to a famous conjecture of Emil Artin. An integer $a$ is a \emph{primitive root} of a number $n$ if $a$ generates multiplicatively the group $\mathbb{Z}_p^*$, namely if $\{a^i \mod p|~i\in\mathbb{Z}\}=\{1,\dots,p-1\}$. 

\begin{conjecture}\label{conj:Artin}
Let $a$ be an integer that is not a square number and not $-1$. The set of $a$-rooted primes is infinite.    
\end{conjecture}

This happens because, in the binary case, a word that is a \BWT image and is completely unclustered is necessarily of the form $(10)^n$. Reciprocally, the word $(10)^n$ is a \BWT image if and only if $2n+1$ is an odd prime and $2$ generates the cyclic group $\mathbb{Z}_{2n+1}^*$, as stated in Theorem \ref{thm:binary}. This is no longer the case if the alphabet contains at least $3$ letters, since more words have maximal $\rle$. However, one can still find a trace of this phenomenon in the larger alphabet case:

\begin{theorem}\label{thm:artin}
Let $k\ge 2$, and $\alpha=k-1\dots \str{0}$, the word obtained by concatenating each letter of $\Sigma_k$ in decreasing lexicographical order. For every integer $n$, the word $\alpha^n$ is a \BWT image if and only if $kn+1$ is a prime and $k$ is a primitive root modulo $kn+1$.
\end{theorem}
\begin{proof}
The proof is essentially the same as in the binary case~\cite{mantaci_measuring_2017}. By relabeling and considering the standard permutation $\pi_{\alpha^n}$ on the set $\{1,\dots kn\}$, one can express $\pi_{\alpha^n}(i)=ki \mod kn+1$, hence $\pi_{\alpha^n}^i(1)=k^i \mod kn+1$. This is a single cycle if and only if $\{k^i \mod kn+1|~i\in\mathbb{Z}\}=\{1,\dots,kn\}$. 

Since $k$ and $kn+1$ are coprime, the left-hand set is contained in $\mathbb{Z}_{kn+1}^*$,
and since $\mathbb{Z}_{kn+1}^*\subseteq\{1,\dots,kn\}$, the three sets are equal.
In particular, $|\mathbb{Z}_{kn+1}^*| = kn$, which holds if and only if $kn+1$ is prime.
Thus $kn+1$ is prime and $k$ is a primitive root modulo $kn+1$.
\end{proof}

\begin{remark}
    Note that, when $\alpha^n$ is a \BWT image, the necklace $u$ such that $\BWT(u)=\alpha^n$ is, by definition, a generalized de Bruijn word (see also Example~\ref{ex:stand-perm} for the case $k=3$ and $n=6$). However, it is never an ordinary de Bruijn word (equivalently, an integer $n$ such that $kn=k^m$ for some integer $m$ never satisfies the condition of Theorem~\ref{thm:artin}). This was shown in~\cite{mantaci_measuring_2017} for the case $k=2$. For $k>2$, observe that $k^m\equiv -1 \mod k^m+1$, so $k^{2m}\equiv 1 \mod k^m+1$. This means that the order of $k\mod k^m+1$ must divide $2m$. But $2m\le 3^m \le  k^m$ for every $m\ge 1$, so $k$ cannot have order $k^m$, and $k^m+1$ cannot have $k$ as a primitive root. 
\end{remark}

\bibliography{references}

\end{document}